\documentclass[journal,draftcls,onecolumn,12pt,twoside]{IEEEtranTCOM}
\normalsize

\ifCLASSINFOpdf
\else
\fi

\usepackage{latexsym}
\usepackage{stmaryrd}
\usepackage{graphicx, amssymb, amsmath}
\usepackage{cite}
\usepackage{amsthm}
\usepackage{verbatim}
\usepackage{algorithm}
\usepackage{algorithmic}

\begin{document}
%
\title{Optimal Power Allocation for Energy Harvesting and Power Grid Coexisting Wireless Communication Systems}
%
%
%

\author{Jie Gong,~\IEEEmembership{Student Member,~IEEE},
Sheng Zhou,~\IEEEmembership{Member,~IEEE}, Zhisheng
Niu,~\IEEEmembership{Fellow,~IEEE}
\thanks{The authors are with Tsinghua National Laboratory for Information Science and
Technology, Department of Electronic Engineering, Tsinghua
University, Beijing 100084, P.~R.~China. Email:
gongj08@mails.tsinghua.edu.cn, \{sheng.zhou, niuzhs\}@tsinghua.edu.cn}
\thanks{The research work is partially sponsored by the National Basic
Research Program of China (973Green: No.~2012CB316001); by the National Natural
Science Foundation of China (No.~61201191, No.~61021001, No.~60925002); and
by Hitachi R\&D Headquarter.}}

%
%

\markboth{IEEE Transactions on Communications}%
{Submitted paper}
%



\maketitle

\begin{abstract}
This paper considers the power allocation of a single-link wireless communication with joint energy harvesting and grid power supply. We formulate the problem as minimizing the grid power consumption with random energy and data arrival, and analyze the structure of the optimal power allocation policy in some special cases. For the case that all the packets are arrived before transmission, it is a dual problem of throughput maximization, and the optimal solution is found by the \emph{two-stage water filling (WF)} policy, which allocates the harvested energy in the first stage, and then allocates the power grid energy in the second stage. For the random data arrival case, we first assume grid energy or harvested energy supply only, and then combine the results to obtain the optimal structure of the coexisting system. Specifically, the \emph{reverse multi-stage WF} policy is proposed to achieve the optimal power allocation when the battery capacity is infinite. Finally, some heuristic online schemes are proposed, of which the performance is evaluated by numerical simulations.
\end{abstract}

\begin{IEEEkeywords}
Energy harvesting and power grid coexisting, convex optimization, two-stage water filling, reverse multi-stage water filling.
\end{IEEEkeywords}

%
\IEEEpeerreviewmaketitle

\section{Introduction}
%
%
%
%
\IEEEPARstart{M}{inimizing} the network energy consumption has become one of the key design
requirements of wireless communication systems in recent years \cite{Marsan:2009, Oh:2011}. By
exploiting the renewable energy such as solar power, wind power and
so on, energy harvesting technology can efficiently reduce CO${}_2$
emissions \cite{David:2010}. It can greatly reduce the load of power grid when the renewable energy usage well meets the network traffic load. However, due to the randomness of the renewable
energy, users' quality of service (QoS) may not be guaranteed all
the time, which requires other complementary stable power
supplies. As the power grid is capable of providing persistent power
input, the coexistence of energy harvesting and grid power supply is
considered as a promising technology to tackle the problem of simultaneously
guaranteeing the users' QoS and minimizing the power grid energy consumption \cite{Cui:2012}.

A number of research paper have focused on the energy harvesting issues
recently. The online power allocation policies of energy
harvesting transmitter with a rechargeable battery are studied in
\cite{Fu:2003, Gatz:2010, Shar:2010, Ho:2010} using Markov decision
process (MDP) approach \cite{Dimi:2005}. The energy allocation and
admission control problem is studied in \cite{Fu:2003} for
communications satellites. Ref. \cite{Gatz:2010} considers the
cross-layer resource allocation problem to maximize the total system
utility. In \cite{Shar:2010}, the authors develop the energy
management policies with stability guarantee in sensor networks. And
throughput maximization with causal side information and that with
full side information are both studied in \cite{Ho:2010}.
Although some optimal online policies are proposed, the structure of
the optimal power allocation is still not clearly presented via MDP
approach, which in addition usually suffers from the \emph{curse of
dimensionality} \cite{Dimi:2005}. Recently, there have been research efforts on the
analysis of the optimal offline power allocation structure
\cite{Jing:2012, Tut:2012, Ozel:2011}. In \cite{Jing:2012}, the problem of
minimizing the transmission completion time with infinite battery
capacity in non-fading channel is studied in two scenarios, i.e.,
all packets are ready before transmission and packets arrive during
transmission. Ref.~\cite{Tut:2012} finds the optimal
transmission policy to maximize the short-term throughput with
limited energy storage capacity, and exploits the relation between the throughput maximization and the transmission completion time minimization. The power allocation of energy harvesting systems in fading
channel is formulated as a convex optimization problem in
\cite{Ozel:2011}, and the optimal \emph{directional
water-filling (WF)} policy is introduced.

Besides the analysis of energy harvesting system listed above, there
is some other work on offline power allocation, which optimizes the
energy efficiency assuming that the required energy is always available (e.g.
power grid). In \cite{Biyik:2002}, a lazy
packet scheduling policy is proved optimal to minimize the total
energy consumption within a deadline. The
similar problem with individual packet delay constraint is studied
in \cite{Chen:2007}. The authors in \cite{Zafer:2009} considered the
strict QoS constraints such as individual packet deadlines, finite
buffer and so on, and proposed a general calculus approach. Different from the existing work, we consider the power allocation problem in the wireless fading channel with the coexistence of energy harvesting and power grid.
To the best of our knowledge, very limited work such as \cite{Cui:2012} has focused on this topic, which however, only considers the binary power allocation.

In particular, we study the offline power allocation problem aiming
to find the optimal structure in the energy harvesting and power
grid coexisting system. We assume that the data is randomly arrived in each transmission
frame, and minimize the average power grid energy
consumption while completing the required data transmission before a given deadline. The
main contributions are presented as follows:

\begin{itemize}
 \item[$\bullet$] The offline grid power minimization problem is formulated and then converted into a convex optimization problem. Then the closed-form optimal power is obtained in terms of Lagrangian multipliers.
 \item[$\bullet$] The structure of the optimal power allocation is analyzed in some special cases. Specifically, in the case that all the data is ready before transmission, the problem is a dual problem of the throughput
 maximization and is solved by the \emph{two-stage WF} algorithm. In the
 case that the battery capacity is infinite, the optimal water levels
 are non-decreasing, and the optimal power allocation structure can be
 obtained by the \emph{reverse multi-stage WF} policy, which
 allocates the harvested energy and the grid energy one by one from the last frame to the beginning.
 \item[$\bullet$] Some sub-optimal online schemes based on the analysis of the optimal
 power allocation structure, including \emph{constant water level}
 policy and \emph{adaptive water level} policy, are proposed and evaluated by numerical simulations. It is shown that the adaptive water level policy performs better than the constant one considering both the grid energy consumption and the QoS guarantee.
\end{itemize}

The rest of the paper is organized as follows. Section
\ref{sec:model} introduces the model of the energy harvesting and
power grid coexisting system. The offline grid power minimization problem is formulated in Section \ref{sec:gpmin}. In Section \ref{sec:struct}, we study the optimal solution for the minimization problem by cases. Then some online algorithms are proposed in Section \ref{sec:online}, and numerical results are presented in Section \ref{sec:simu}. Finally, Section \ref{sec:conc} concludes the paper.

\begin{figure}
\centering
\includegraphics[width=5in]{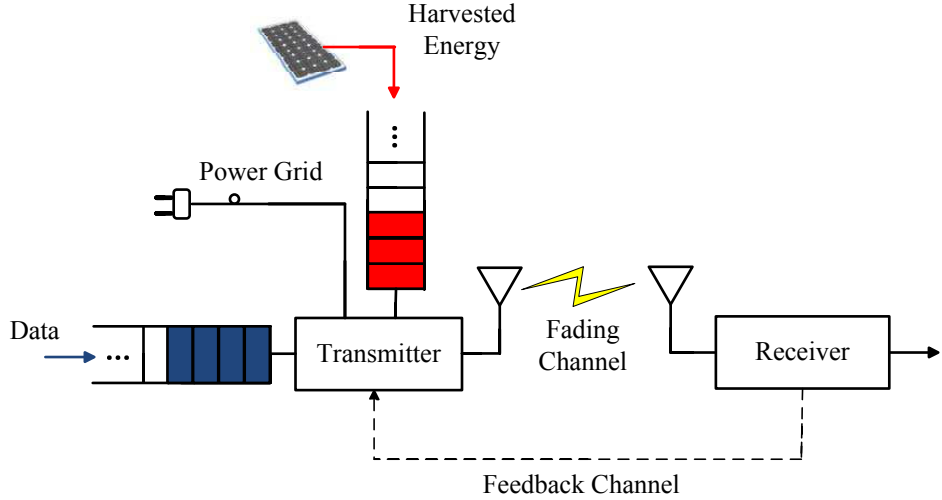}
\caption{Single-link wireless communication model with the coexistence of energy
harvesting and power grid.} \label{fig:system}
\end{figure}

\section{System Model} \label{sec:model}
We consider a single wireless link as shown in Fig.
\ref{fig:system}, where the transmitter is powered by both energy
harvesting and power grid. The harvested energy is stored in a
battery with capacity $E_{\mathrm{max}}$ for data
transmission usage. If the energy in the battery is insufficient, the
transmitter uses the power grid to guarantee the user's QoS. In this paper, we assume the energy is used only
for transmission, i.e, the processing energy is ignored.

Assume the system is slotted with frame length $T_f$. We study the optimal power allocation during a finite number of frames $N$, indexed by $\{1, \ldots, N\}$. We adopt the block fading model, i.e, the channel keeps
constant during each frame, but varies from frame to frame. The received signal in frame $i \in \{1, \ldots, N\}$ is given by
\begin{equation}
y_i = h_i\sqrt{p_i}x_i + n_i,
\end{equation}
where $h_i$ is the channel gain, $p_i$ is the transmit power, $x_i$ is the input symbol with unit norm, and
$n_i$ is the additive white Gaussian noise (AWGN) with zero mean
and unit variance. The reference signal-to-noise ratio (SNR), defined as the SNR with unit transmit power, during the frame $i$ is denoted by $\gamma_i = |h_i|^2$. Then, the instantaneous rate $r_i$ in bits per channel
use is
\begin{equation}
r_i = \frac{1}{2}\log_2(1+\gamma_ip_i). \label{eqn:rate}
\end{equation}

Assume the initial battery energy is $E_0$, and the harvested energy $E_i, i = 1, \ldots, N-1$ arrives at the beginning of each frame $i$. Part of the transmit power $p_i$ is supplied by the battery, denoted as $p_{H,i}$. And the rest is supplied by the power grid $p_{G,i} = p_i-p_{H,i}$. The frame and energy arrival model is depicted in Fig. \ref{fig:timeline}.

\begin{figure}
\centering
\includegraphics[width=5in]{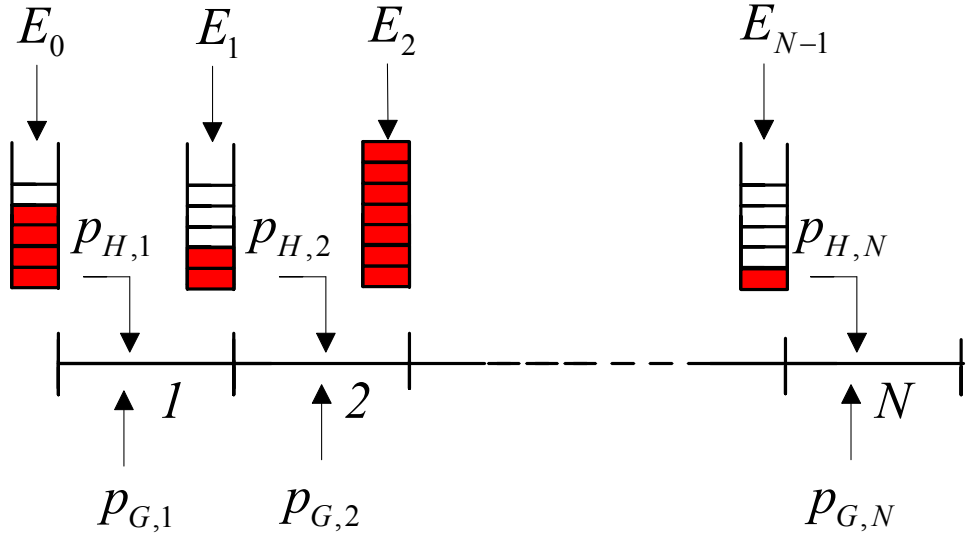}
\caption{Frame model and the energy arrival and storage structure.} \label{fig:timeline}
\end{figure}

There are several constraints on the battery power $p_{H,i}$: Since
the harvested energy cannot be consumed before its arrival, the
first constraint is the energy causality
\begin{equation}
\sum\limits_{i=1}^{k}T_fp_{H,i} \le \sum\limits_{i=0}^{k-1}E_i, \quad
k = 1, 2, \ldots, N. \label{cnstr:causal}
\end{equation}
The second constraint is the limited battery capacity. The battery overflow happens when the reserved energy
plus the harvested energy exceeds the battery capacity, which however, is not preferred because the data rate can be increased if the energy is used in advance instead of overflowed. So we have
\begin{equation}
\sum\limits_{i=0}^{k}E_i - \sum\limits_{i=1}^{k}T_fp_{H,i} \le
E_{\mathrm{max}}, \quad k = 1, 2, \ldots, N-1. \label{cnstr:Ecap}
\end{equation}

On the other hand, the grid power satisfies the average power constraint
\begin{equation}
\frac{1}{N}\sum\limits_{i=1}^{N}p_{G,i} =
\frac{1}{N}\sum\limits_{i=1}^{N}(p_i - p_{H,i}) \le
P_{G,\mathrm{ave}}, \label{cnstr:Pave}
\end{equation}
where $P_{G, \mathrm{ave}}$ is the average grid power supply.

\section{Grid Power Minimization Problem} \label{sec:gpmin}

We study the grid power minimization problem with offline setup, i.e., the energy arrivals $E_i$ and the channel gains $\gamma_i$
of all the frames are known in advance. Assume that the amount of $B_i$ bits arrive at the beginning of
frame $i+1$, where $i = 0, 1, \ldots, N-1$, and all the quantities are known. The constraint due to
the random packet arrival is that the number of transmitted bits can
not exceed those in the buffer, which can be expressed as
\begin{eqnarray}
\sum\limits_{i=1}^k \frac{T_f}{2} \log_2(1 + \gamma_ip_i)  \le
\sum\limits_{i=0}^{k-1} B_i, & k = 1, \ldots, N-1 \label{constr:RanBit1}\\
\sum\limits_{i=1}^N \frac{T_f}{2} \log_2(1 + \gamma_ip_i) =
\sum\limits_{i=0}^{N-1} B_i.&{}\label{constr:RanBit2}
\end{eqnarray}
By using (\ref{constr:RanBit2}) minus (\ref{constr:RanBit1}) for all
$k = 1, \ldots, N-1$, and then relaxing (\ref{constr:RanBit2}) to
inequality, we get the data constraint as
\begin{equation}
\sum\limits_{i=k}^N \frac{T_f}{2} \log_2(1 + \gamma_ip_i) \ge
\sum\limits_{i=k-1}^{N-1} B_i, \quad k = 1, \ldots, N
\label{constr:RanBit}
\end{equation}
which is equivalent with the original constraints (\ref{constr:RanBit1}) and (\ref{constr:RanBit2}) because while
achieving the optimality, (\ref{constr:RanBit}) for $k = 1$ must be
satisfied with equality. Otherwise, we can always decrease the
transmission rate by reducing the energy consumption without
conflicting any other constraints. In addition, the relaxation of (\ref{constr:RanBit2}) avoids the conflicts between power and data constraints. For instance, if $E_0 = E_1 =
E_{\mathrm{max}}$ and $B_0 = 0$, the original constraint
(\ref{constr:RanBit1}) conflicts with the battery constraint
(\ref{cnstr:Ecap}) for $k=1$. In this case, the problem is infeasible. However, the problem with constraint (\ref{constr:RanBit})
is always feasible since we can use the extra harvested
energy to transmit packets carrying no data to avoid battery overflow. In the example, the feasible power allocation policy is to allocate $E_0$ to frame 1 to transmit packets without any data.

We aim at minimizing the energy consumption of power grid under the
constraints of required data transmission, causality of harvested
energy, battery capacity and non-negative power allocation. The
problem is formulated as follows
\begin{eqnarray}
&\mathop{\min}\limits_{\begin{subarray}{c} p_i, p_{H,i} \\ i = 1, \ldots, N \end{subarray}} & \sum\limits_{i=1}^{N}T_f(p_i-p_{H,i}) \label{opt:pmin}\\
&\textrm{s.t.} &  (\ref{cnstr:causal}), (\ref{cnstr:Ecap}), \mathrm{~and~} (\ref{constr:RanBit}), \nonumber\\
& & p_{H,k} \le p_k , \quad \forall k \label{cnstr:pG}\\
& & p_{H,k} \ge 0 , \quad \forall k \label{cnstr:pH}
\end{eqnarray}
where the constraints (\ref{cnstr:pG}) and (\ref{cnstr:pH}) imply
that the energy from either power grid or battery is non-negative.
Notice that there are three cases of the solutions: (a) The harvested energy exceeds the required energy to transmit all the arrived packets, which results in $p_i = p_{H,i}$, and there are multiple solutions; (b) The harvested energy is insufficient to transmit the packets. Then we have $p_i \ge p_{H,i}$, and the inequality is satisfied for at least one $i$; (c) The harvested energy is just enough to achieve the optimal power allocation, and any additional packets to transmit will result in non-zero grid power. In our analysis, we mainly focus on the case (b) where grid power is necessary. The case (c) is also studied as it helps us to understand the structure of the optimal harvested power allocation.

Also note that we ignore the average grid power constraint (\ref{cnstr:Pave}), which decides the feasibility of the problem. If we can find a feasible solution for problem (\ref{opt:pmin}) without violating (\ref{cnstr:Pave}), it can be ignored. Otherwise, the required bits can not be transmitted by the deadline. Hence, the best we can do is try to maximize the number of transmitted bits, which can be formulated as a throughput maximization problem
\begin{eqnarray}
&\mathop{\max}\limits_{\begin{subarray}{c} p_i, p_{H,i} \\ i = 1, \ldots, N \end{subarray}} & \sum\limits_{i=1}^{N} \frac{T_f}{2}
\log_2(1+\gamma_ip_i) \label{opt:thrmax}\\
&\textrm{s.t.} & (\ref{cnstr:causal}), (\ref{cnstr:Ecap}), (\ref{cnstr:Pave}), (\ref{cnstr:pG}), \mathrm{~and~} (\ref{cnstr:pH}). \nonumber
\end{eqnarray}
In the following section, we analyze the optimal solution for the grid power minimization problem assuming it is feasible (the constraint (\ref{cnstr:Pave}) is ignored). Also, the relation between the throughput maximization and the grid power minimization is briefly discussed.

\section{Offline Optimal Solution Structure} \label{sec:struct}
The objective (\ref{opt:pmin}) is a linear function, the constraints
(\ref{constr:RanBit}) are convex for all $k$ since it is a sum of
$\log$ functions larger than a threshold, and the others are all
linear constraints. Consequently, the optimization
problem is a convex optimization problem, and the
optimal solution satisfies the KKT conditions \cite{Boyd:2004}. Define the Lagrangian
function for any multipliers $\lambda_k \ge 0, \mu_k \ge 0,\eta_k \ge 0, \alpha_k
\ge 0, \beta_k \ge 0$ as
\begin{eqnarray}
\mathcal{L} & = & \sum\limits_{i=1}^{N}T_f(p_i-p_{H,i}) \nonumber\\
& & +\sum_{k=1}^{N}\lambda_k\left(\sum\limits_{i=1}^{k}T_fp_{H,i}
- \sum\limits_{i=0}^{k-1}E_i\right) \nonumber\\
& & +\sum_{k=1}^{N-1}\mu_k\left(\sum\limits_{i=0}^{k}E_i -
\sum\limits_{i=1}^{k}T_fp_{H,i} - E_{\mathrm{max}}\right) \nonumber\\
& & -\sum\limits_{k=1}^N \eta_k \left(\sum\limits_{i=k}^N
\frac{T_f}{2} \log_2(1 + \gamma_ip_i)
-\sum\limits_{i=k-1}^{N-1} B_i \right) \nonumber\\
& & +\sum_{k=1}^{N}\alpha_k(p_{H,k} - p_k) -\sum_{k=1}^{N}\beta_kp_{H,k} \label{eqn:lagr2}
\end{eqnarray}
with additional complementary slackness conditions
\begin{eqnarray}
\lambda_k\left(\sum\limits_{i=1}^{k}T_fp_{H,i} - \sum\limits_{i=0}^{k-1}E_i\right) = 0,
& & \forall k \label{comp:causal}\\
\mu_k\left(\sum\limits_{i=0}^{k}E_i -\sum\limits_{i=1}^{k}T_fp_{H,i} - E_{\mathrm{max}}\right) = 0,
& & k < N \label{comp:Ecap}\\
\eta_k \left(\sum\limits_{i=k}^N \frac{T_f}{2} \log_2(1 + \gamma_ip_i)
- \sum\limits_{i=k-1}^{N-1} B_i \right) = 0. & & \forall k
\label{comp:Bits}\\
\alpha_k(p_{H,k} - p_k) = 0, & & \forall k \label{comp:pG}\\
\beta_kp_{H,k} = 0, & & \forall k \label{comp:pH}
\end{eqnarray}

We apply the KKT optimality conditions to the Lagrangian function
(\ref{eqn:lagr2}). By setting $\partial\mathcal{L}/\partial p_i = \partial\mathcal{L}/\partial p_{H,i} =0$, we obtain the unique optimal power levels $p_i^*$ in terms of the Lagrange multipliers as
\begin{eqnarray}
p_i^* & = & \frac{\sum_{k=1}^i\eta_k}{1-\alpha_i}-\frac{1}{\gamma_i}, \label{opt:pran1}\\
\textrm{or}\; p_i^* & = & \frac{ \sum_{k=1}^i\eta_k }
{\sum_{k=i}^{N}(\lambda_k-\mu_k)-\beta_i}-\frac{1}{\gamma_i},
\label{opt:pran2}
\end{eqnarray}
where $\mu_N = 0$. Notice that there is no closed-form expression for the parameter $p_{H,i}$, which indicates that the power supplied from battery can be of many choices as long as the constraints are satisfied and the optimal total power $p_i^*$ is achieved. For instance, if one of the optimal power allocation for frames $i$ and $i+1$ is $p_i= p_{i+1} = 2, p_{H,i}=1.5, p_{H,i+1}=0.5$, then it is also optimal by allocating $p_{H,i}= p_{H,i+1}=1$. Denote $p_{H,i}^*$ as a feasible optimal solution in order to distinguish from the optimization parameter $p_{H,i}$. Hence, the optimal power allocation is
\begin{equation}
p_i^* = \left[\nu_i-\frac{1}{\gamma_i}\right]^+, \label{opt:p}
\end{equation}
where $[x]^+ = \max\{x,0\}$, and the water level is expressed by either
\begin{eqnarray}
\nu_i &=& \sum_{k=1}^{i}\eta_k, \label{wl:1} \\\textrm{or}\; \nu_i
&=& \frac{\sum_{k=1}^{i}\eta_k}{\sum_{k=i}^{N}(\lambda_k-\mu_k)},
\label{wl:2}
\end{eqnarray}
depending on how the battery power and the grid power are allocated.

Based on the optimal power allocation solution
(\ref{opt:p})-(\ref{wl:2}), we study the following special cases to find the structure of
the optimal power allocation.

\theoremstyle{proposition} \newtheorem{proposition}{Proposition}

\subsection{Packets Ready Before Transmission}
Assume the packet arrival sequence is $\{B, 0, \ldots, 0\}$, i.e., all the packets are ready before transmission. In this case, the optimal power allocation is simplified as
\begin{eqnarray}
p_i^* & = & \frac{\eta_1}{1-\alpha_i}-\frac{1}{\gamma_i}, \label{opt:pktready1}\\
\textrm{or}\; p_i^* & = & \frac{ \eta_1 }
{\sum_{k=i}^{N}(\lambda_k-\mu_k)-\beta_i}-\frac{1}{\gamma_i},
\label{opt:pktready2}
\end{eqnarray}

We have the following proposition:
\begin{proposition} \label{prop:noRandom}
For given harvested energy profile $E_{H,i}, i = 1, \ldots, N$, define $f(P_{G,\mathrm{ave}}) = \max$ $\sum_{i=1}^{N} \frac{T_f}{2} \log(1+\gamma_ip_i)$ under the constraints (\ref{cnstr:causal}), (\ref{cnstr:Ecap}), (\ref{cnstr:Pave}), (\ref{cnstr:pG}), and (\ref{cnstr:pH}). Assume the packet arrival sequence is $\{B, 0, \ldots, 0\}$, and define $g(B) = \min\frac{1}{N} \sum_{i=1}^{N} T_f(p_i - p_{H,i})$ under the constraints (\ref{cnstr:causal}), (\ref{cnstr:Ecap}), (\ref{constr:RanBit}), (\ref{cnstr:pG}), and (\ref{cnstr:pH}). We have $f = g^{-1}$.
\end{proposition}

\begin{proof}
First of all, the functions $f$ and $g$ are monotonic increasing since the $\log$ functions are monotonically increasing and concave.

Then, we set $B_{\mathrm{max}}=f(P),
P_{\mathrm{min}}=g(B_{\mathrm{max}})$. If $P_{\mathrm{min}} > P$, we
have $f(P_{\mathrm{min}}) > f(P) = B_{\mathrm{max}}$. Hence, $g(B_{\mathrm{max}}) < P_{\mathrm{min}}$, as with $P_{\mathrm{min}}$, we can achieve higher throughput $f(P_{\mathrm{min}})$ by the optimal policy. It
contradicts the assumption. The same result holds for
$P_{\mathrm{min}} < P$. As a result, $g(B_{\mathrm{max}}) =
P_{\mathrm{min}} = P = f^{-1}(B_{\mathrm{max}})$. 
\end{proof}

{Proposition \ref{prop:noRandom}} shows that if all the packets are already at the transmitter before the transmission starts, the problem (\ref{opt:pmin}) is a dual problem of (\ref{opt:thrmax}). We now discuss the structure of the optimal power allocation of problem (\ref{opt:thrmax}), and our grid power minimization problem can be solved in the same way due to dual property. In fact, the problem (\ref{opt:thrmax}) can be considered as a simple extension of \cite{Ozel:2011} by adding an additional average grid power constraint. Hence, it can be solved similarly by the directional WF algorithm. In particular, the harvested energy in each frame can only be transferred to the future frames, and the amount that can be transferred is limited by the battery capacity. It is realized by introducing the concept of right permeable tap \cite{Ozel:2011}. The difference in our problem is that there is additional grid energy, which however, can be considered as an additional battery with total amount of energy $NP_{G,\mathrm{ave}}T_f$. As there are no causality and capacity constraints, the grid energy allocation can be viewed as the directional WF from frame 1 with right permeable tap which allows any amount of energy transfer, which is in fact the traditional WF algorithm \cite{Gold:2005}.

In order to emphasize the difference between harvested energy and grid energy as well as motivate the algorithm design for random packet arrival case, we propose the \emph{two-stage WF} optimal power allocation policy. The procedure is divided into two stages. In the first stage, the system operates directional WF
\cite{Ozel:2011} to allocate the harvested energy. At the beginning of
the second stage, the surface of the water freezes up. Then the
traditional WF \cite{Gold:2005} is performed to allocate the power grid energy
on the ice surface, until the energy is used up.

\begin{figure}
\centering
\includegraphics[width=4.5in]{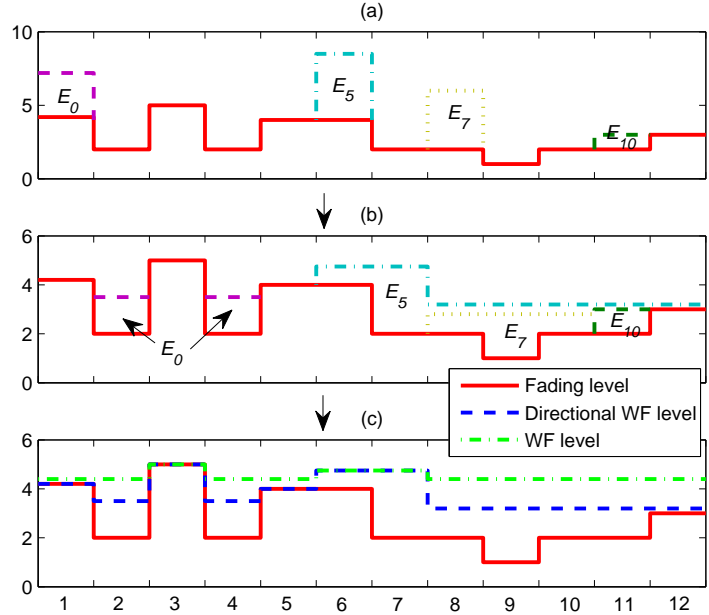}
\caption{Example of two-stage WF operation. (a) shows the harvested energy arrival; (b) shows the procedure and the result of directional WF; and (c) shows the two-stage WF result.} \label{fig:WFex}
\end{figure}

An example is shown in Fig. \ref{fig:WFex}. We illustrate a total of $N=12$ frames. The initial
battery energy is $E_0$, and only the values of harvested energy $E_5, E_7$ and $E_{10}$ is non-zero. By performing the two-stage WF
algorithm, the optimal power allocation is obtained. It is observed
from the figure that a constant water level is achieved except for frames 3, 6 and 7. The power used in
frame 3 is zero because of deep fading. Due to the battery capacity
constraint, frames 6 and 7 have higher water level to avoid energy
overflow. Also, the power allocation of harvested energy and power
grid energy is variable under the water level. For instance, the
harvested energy allocated in frame 2 can be partially reserved to
be used in frame 4, and the water level in frame 2 is achieved by
power grid energy.

\subsection{Random Packet Arrival}
Now assume the packets arrive in every frame. In this condition, we first analyze the structure of the optimal power allocation under the different assumptions on the energy. Then based on the observations, we discuss the power allocation for the general case.

\begin{proposition} \label{prop:noEH}
If $E_i = 0$ for all $i = 0, 1, \ldots, N-1$, the optimal water
levels are monotonically non-decreasing as $\nu_{k} \le \nu_{k+1}$.
The water level may increase only when (\ref{constr:RanBit1}) is
satisfied with equality.
\end{proposition}
\begin{proof}
Without the energy harvesting, the water
levels are expressed as  (\ref{wl:1}). Since $\eta_{k+1} \ge 0$,
it naturally results in $\nu_{k} \le \nu_{k+1}$. In addition, if
(\ref{constr:RanBit1}) is satisfied without equality, i.e.,
\begin{equation}
\sum_{i=1}^k \frac{T_f}{2} \log_2(1 + \gamma_ip_i) < \sum_{i=0}^{k-1}
B_i,
\end{equation}
from  (\ref{constr:RanBit2}) we get
\begin{equation}
\sum_{i=k+1}^N \frac{T_f}{2} \log_2(1 + \gamma_ip_i) >
\sum_{i=k}^{N-1} B_i.
\end{equation}
Then from (\ref{comp:Bits}) we have $\eta_{k+1} = 0$, and hence
$\nu_{k} = \nu_{k+1}$. 
\end{proof}

Note that the energy minimization by a deadline with grid power supply only is studied in \cite{Biyik:2002} for AWGN channel, where the rate-power function is time-invariant. Here, we consider the fading channel case, and conclude from {Proposition \ref{prop:noEH}} that if the power is supplied by the power grid only, the water level is
non-decreasing from frame to frame. It increases only at the points
where
\begin{equation}
\sum_{i=1}^k \frac{T_f}{2} \log_2(1 + \gamma_ip_i) = \sum_{i=0}^{k-1}
B_i, \label{cond:equal}
\end{equation}
i.e., the already arrived packets are all transmitted by the end of
frame $k$. The optimal power allocation can be obtained by \emph{reverse multi-stage WF} algorithm, which allocates the power from the last frame to the first frame. First allocate the power to frame $N$ so that
(\ref{constr:RanBit}) for $k=N$ is satisfied with equality. Then
freeze up the water surface of frame $N$, and allocate the power to
frames $N-1$ and $N$ with the traditional WF policy on the ice surface so that
(\ref{constr:RanBit}) for $k=N-1$ is satisfied with equality. Repeat
the procedure until the first frame. We summarize it as
{Algorithm \ref{alg:revPA}}.

\begin{algorithm}[th]
\caption{Reverse multi-stage WF with grid power} \label{alg:revPA}
\begin{algorithmic}[1]

%

\STATE Set $\bar{\gamma}_i = \gamma_i, i = 1, \ldots, N.$

\FORALL{$k = N$ to $1$}

\STATE Find $\nu_k$ so that
\begin{equation}
\sum\limits_{i=k}^N \frac{T_f}{2}
\log_2(1 + \bar{\gamma}_ip_i) = {B}_{k-1}, p_i =
\left[\nu_k-\frac{1}{\bar{\gamma}_i}\right]^+. \label{eq:stageWF}
\end{equation}

\STATE Update $\bar{\gamma}_i$ as
\begin{equation}
\bar{\gamma}_i = \left(p_i + \frac{1}{\bar{\gamma}_i}\right)^{-1}, i
= k, \ldots, N. \label{gamma:update}
\end{equation}

\ENDFOR

\STATE $p_i^* = \frac{1}{\bar{\gamma}_i}- \frac{1}{{\gamma}_i}, i =
1, \ldots, N.$

\end{algorithmic}
\end{algorithm}


In the algorithm, for a given $k$, $\bar{\gamma}_i$ denotes the freezed water level updated in the previous stage $i-1, i = k+1, \ldots N$, and $\bar{\gamma}_k = {\gamma}_k$ as was set in step 1. Hence, (\ref{eq:stageWF}) means that in this stage, $B_{k-1}$ bits arrived at the beginning of frame $k$ are served by the traditional water filling algorithm in the following frames (from $k$ to $N$). Fig. \ref{fig:revPAex} demonstrates an example of reverse multi-stage WF with grid power. It can be seen the arrived bits $B_2$ and $B_3$ are served through frames 3-5, resulting in the same water level in these frames. The rest bits are served in their arriving frames respectively.

\begin{figure}
\centering
\includegraphics[width=4.5in]{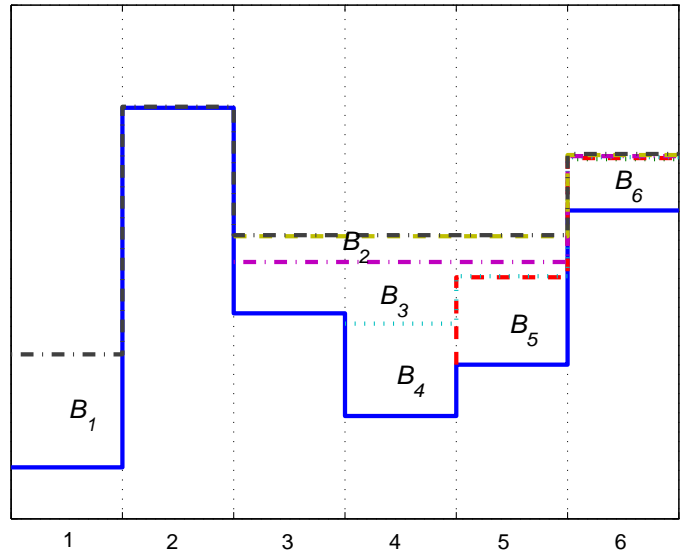}
\caption{Example of reverse multi-stage WF with grid power. Starting from the last frame to the first one, $B_i$ is served in frames $i, i+1, \ldots, N$ with traditional WF scheme.} \label{fig:revPAex}
\end{figure}

\begin{proposition} \label{prop:noEmax}
If $E_{\mathrm{max}} = \infty$, $p_i^* = p_{H,i}^*$ for all $i =
1, 2, \ldots, N$, and $p_{H,i}^*$ is unique (i.e., any packet arrival set $\{B_0', \ldots, B_{N-1}'\}$ satisfying $B_i' \ge B_i$, and $\sum_{i=0}^{N-1}B_i' > \sum_{i=0}^{N-1}B_i$ results in $\sum_{i=1}^{N}T_f(p_i^* -p_{H,i}^*) > 0$), the optimal water levels are monotonically
non-decreasing as $\nu_{k} \le \nu_{k+1}$. The water level may
increase when either (\ref{cnstr:causal}) or (\ref{constr:RanBit1})
is satisfied with equality.
\end{proposition}
\begin{proof}
Under the conditions that the energy
battery has infinite capacity and there is a unique harvested power allocation policy to complete the transmission without grid power input, the water level is determined by
\begin{equation}
\nu_k = \frac{\sum_{j=1}^{k}\eta_j}{\sum_{j=k}^{N}\lambda_j}.
\label{wl:noEmax}
\end{equation}
Since $\eta_{k+1} \ge 0, \lambda_{k} \ge 0$, we have $\nu_{k} \le
\nu_{k+1}$.

The water level increases when either $\eta_{k+1} > 0$ or
$\lambda_{k} > 0$. Based on (\ref{comp:causal}), $\lambda_{k} > 0$
indicates that (\ref{cnstr:causal}) is satisfied with equality.
Based on (\ref{constr:RanBit2}) and (\ref{comp:Bits}), $\eta_{k+1} >
0$ indicates that (\ref{constr:RanBit1}) is satisfied with equality.
\end{proof}

Proposition \ref{prop:noEmax} shows the structure of optimal harvested power allocation under infinite battery capacity constraint. The conclusion is an extension of the optimal power allocation in non-fading scenario ({Lemma 6} in \cite{Jing:2012}) to the fading scenario. The intuition is that due to the causality of energy and data arrival, the later frames can ``see" more available energy to be used and data bits to be transmitted than the earlier frames as the energy and data can be transferred rightwards, rather than leftwards. Hence, the water level is non-decreasing. Without the battery capacity constraint, the harvested energy can be
reserved for future use in any way as one wishes, and the arrived
packets can also be transmitted in future frames. In addition, the water level
changes when either all the packets arrived before the current frame
are transmitted or all the energy harvested before the current frame
are used up. If both the packet buffer and the energy battery are
not empty, the water level keeps constant.

\begin{proposition} \label{prop:ConoEmax}
If $E_{\mathrm{max}} = \infty$, the optimal water levels are
monotonically non-decreasing as $\nu_{k} \le \nu_{k+1}$. In
addition, 
there is an optimal policy where a
unique $\bar{k}$ exists such that for $k \le \bar{k}$, $\nu_k$ is expressed
as (\ref{wl:1}) and for $k
> \bar{k}$, $\nu_k$ is expressed as (\ref{wl:noEmax}).
\end{proposition}
\begin{proof}
We prove the proposition by studying all the cases
of power allocation between energy harvesting and power grid in two
consecutive frames.

(a) If $p_k^* > p_{H,k}^{*}$ and $p_{k+1}^* > p_{H,k+1}^{*}$, the water
levels of frame $k$ and $k+1$ have the same expression as
(\ref{wl:1}), hence $\nu_{k} \le \nu_{k+1}$.

(b) If $p_k^* = p_{H,k}^{*}$ and $p_{k+1}^* = p_{H,k+1}^{*}$, the water
levels are expressed the same as (\ref{wl:noEmax}). We also obtain
the monotonicity.

(c) If $p_k^* > p_{H,k}^{*}$ and $p_{k+1}^* = p_{H,k+1}^{*}$, we have
$\alpha_k = 0, \alpha_{k+1} \ge 0$. As a result, $\nu_k = \sum_{j =
1}^k\eta_j \le \sum_{j = 1}^{k+1}\eta_j \le \sum_{j =
1}^{k+1}\eta_j/(1-\alpha_{k+1}) = \nu_{k+1}$.

(d) If $p_k^* = p_{H,k}^{*}$ and $p_{k+1}^* > p_{H,k+1}^{*}$, we can
exchange the amount of harvested energy $T_f\min\{p_{H,k}^{*},$ $
p_{k+1}^* - p_{H,k+1}^{*}\}$ of frame $k$ with the same amount of
grid energy of frame $k+1$ to achieve the same optimal result,
without conflicting any constraint or changing any water level. Consequently, this case is
converted to either (a) or (c), and the monotonicity holds as well.

To sum up, $\nu_{k} \le \nu_{k+1}$ holds for all $k$. In addition,
by conservatively using the harvested energy, the case (d) is
converted to (a) or (c), hence does not exist. Depending on (a)-(c),
there is a unique $\bar{k}$ such that for $k < \bar{k}$, we have
$p_k^* > p_{H,k}^{*}$ and $p_{k+1}^* > p_{H,k+1}^{*}$, for $k >
\bar{k}$, we have $p_k^* = p_{H,k}^{*}$ and $p_{k+1}^* = p_{H,k+1}^{*}$,
and for $k = \bar{k}$, we have $p_k^* > p_{H,k}^{*}$ and $p_{k+1}^* =
p_{H,k+1}^{*}$. Equivalently, for $k \le \bar{k}$, $\nu_k$ is
expressed as (\ref{wl:1}) and for $k > \bar{k}$, $\nu_k$ is
expressed as (\ref{wl:noEmax}). 
\end{proof}

Proposition \ref{prop:ConoEmax} presents an analytically tractable structure of the optimal power allocation for both energy harvesting and power grid. Intuitively, the harvested energy is reserved in the battery for the use in the later frames, in order to reduce the effect of causality constraint and improve the flexibility of harvested power allocation. Hence, the harvested power and the grid power are allocated one by one with the structures described in Proposition \ref{prop:noEH} and \ref{prop:noEmax}, respectively.

The optimal water level can be obtained by the power allocation policy which avoids case (d) in the proof of
{Proposition \ref{prop:ConoEmax}}. It is structured as
follows: The water level is non-decreasing and the harvested energy
is used in a conservative way. That is, the energy harvested before
frame $\bar{k}$ is saved for future use so that the required
transmission after frame $\bar{k}$ can be satisfied by harvested
energy only. To minimize the power grid energy, the harvested energy should be
allocated to maximize the achievable throughput. According to
{Propositions \ref{prop:noEmax}} and
{\ref{prop:ConoEmax}}, the optimal harvested energy
allocation takes the similar reverse calculation structure as
{Algorithm \ref{alg:revPA}}. The difference is that we should
consider the causality of energy arrival. As the harvested energy
cannot be used before its arrival, it should be fully utilized
after arrival. Specifically, denote the available power as $p_{E,N} = E_{N-1}/T_f$, and the required power as $p_{B,N} =
(2^{2B_{N-1}/T_f}-1)/\gamma_N$. If $p_{E,N} < p_{B,N}$, i.e., the
harvested energy in frame $N$ is not enough to transmit $B_{N-1}$
bits, we ``move" the energy from the previous frames until $B_{N-1}$
bits can be transmitted. If $p_{E,N} > p_{B,N}$, i.e., $B_{N-1}$ bits
transmission does not use up the harvested energy $E_{N-1}$, we move
the packets from the previous frames until $E_{N-1}$ is used up.
Then we freeze up the water surface and allocate power to the frames
$N-1$ and $N$ with the traditional WF policy to achieve the same result, i.e., either the bits arrived in this frame are served or the energy arrived is used up. Repeat the
procedure until the harvested energy is used up. It is detailed as {Algorithm
\ref{alg:revPAEH}}.

\begin{algorithm}[th]
\caption{Reverse multi-stage WF with harvested energy}
\label{alg:revPAEH}
\begin{algorithmic}[1]

\STATE Set $\bar{\gamma}_i = \gamma_i, \bar{B}_i = B_i, \bar{E}_i =
E_i, \forall i.$

\FORALL{$k = N$ to $1$}

\STATE Find $\nu_{E,k}$ so that $\sum\limits_{i=k}^NT_fp_{E,i} =
\bar{E}_{k-1}, p_{E,i} =
\left[\nu_{E,k}-\frac{1}{\bar{\gamma}_i}\right]^+.$

\STATE Find $\nu_{B,k}$ so that $\sum\limits_{i=k}^N \frac{T_f}{2}
\log_2(1 + \bar{\gamma}_ip_{B,i}) = \bar{B}_{k-1}, p_{B,i} =
\left[\nu_{B,k}-\frac{1}{\bar{\gamma}_i}\right]^+.$


\IF{$\nu_{E,k} > \nu_{B,k}$}

\STATE Find $j^* = \max\left\{j \Big| \sum\limits_{i=j}^{k-1}
\bar{B}_i \ge \sum\limits_{i=k}^N \frac{T_f}{2} \log_2(1 +
\bar{\gamma}_ip_{E,i})\right\}.$

\STATE Update $\bar{\gamma}_i, i \ge k$ as (\ref{gamma:update}),
$p_i = p_{E,i}$, $\bar{B}_{j^*} = \sum\limits_{i=j^*}^{k-1} \bar{B}_i
- \sum\limits_{i=k}^N \frac{T_f}{2} \log_2(1 + \bar{\gamma}_ip_{E,i}),
\bar{B}_i = 0 , i > j^*.$

\ELSE

\STATE Find $j^* = \max\left\{j \Big| \sum\limits_{i=j}^{k-1}
\bar{E}_i \ge \sum\limits_{i=k}^NT_fp_{B,i}\right\}.$

\IF {$\{j^*\} = \emptyset$}

\STATE Find $\{p_i'\}_{i\ge k} = \arg\max\limits_{\sum_{i = k}^Np_i
\le \sum_{i=0}^{k-1} \bar{E}_i}\frac{T_f}{2} \log_2(1 +
\bar{\gamma}_ip_i).$

\STATE Update $\bar{\gamma}_i, i \ge k$ as (\ref{gamma:update}) with
$p_i = p_i'$, then break.

\ELSE

\STATE Update $\bar{\gamma}_i, i \ge k$ as (\ref{gamma:update}) with
$p_i = p_{B,i}$, $\bar{E}_{j^*} = \sum\limits_{i=j^*}^{k-1} \bar{E}_i
- \sum\limits_{i=k}^NT_fp_{B,i}, \bar{E}_i = 0 , i > j^*.$

\ENDIF

\ENDIF

\ENDFOR

\STATE $p_{H,i}^* = \frac{1}{\bar{\gamma}_i}- \frac{1}{{\gamma}_i}, i
= 1, \ldots, N.$

\end{algorithmic}
\end{algorithm}

Note that we assume the harvested energy is not enough to satisfy
the required transmission, a proper $j^*$ can always be found in
step 6, i.e., we can shift the arrived bits from previous frames to the current one to use up the arrived energy. On the contrary, $j^*$ in step 9 may not be found when $k \le \bar{k}$, and Algorithm \ref{alg:revPAEH} terminates. Then the water surface freezes up and the grid power is allocated using Algorithm \ref{alg:revPA} on the ice surface until all the bits are transmitted. Consequently, the optimal water level is found.

In summary, the power allocation for infinite battery capacity is done in two stages. In the first stage, harvested energy is allocated reversely from last frame by the reverse multi-stage WF algorithm with harvested energy. We call the stage in this algorithm as sub-stage. In each sub-stage, the to-be-transmitted bits form a required water level, and the to-be-used energy can achieve an available water level. If the required water level is higher than the available water level, energy from earlier frames are transferred rightwards. Otherwise, the bits from earlier frames are transferred. The sub-stage ends when the two water levels are equal. When the harvested power allocation stage is finished, the water surface freezes up and the grid energy is allocated in the second stage also in multiple sub-stages. The objective of each sub-stage is just to achieve the required water level.


However, in general case, if the battery capacity is finite, the monotonicity of the
optimal water level does not hold when the constraint
(\ref{cnstr:Ecap}) is satisfied with equality. The structure of the optimal power allocation cannot be described in a simple and clear way. Intuitively, in contrast to the case of infinite battery
capacity where the harvested energy is conservatively used to
achieve the optimality, in the finite capacity case, the power from
battery should be allocated first to minimize the energy waste. We
will propose some online algorithms based on the intuition.

\section{Online Power Allocation} \label{sec:online}

In this section, we assume that the transmitter only has the knowledge of the battery energy levels, the energy arrivals and the channel states of the current frame and all the past ones, but does not know those of the future. Besides, the statistics of energy arrival and fading are available to the transmitter. In this setup, some online power allocation algorithms are proposed.

\subsection{All packets ready before transmission}
If all the packets $B$ are already at the transmitter before
transmission, such as file transfer, the following algorithms are proposed:

\subsubsection{Constant water level}
Motivated by the offline optimal solution that the power allocation
tries to achieve a constant water level (by the traditional WF in the second stage), we can propose a constant
water level algorithm. Assume the constant water level is
$1/\gamma_0$, which is calculated by solving the following equation
\begin{equation}
\int_{\gamma_0}^{\infty}\frac{1}{2} \log_2 \left(
\frac{\gamma}{\gamma_0} \right) f(\gamma) d\gamma = \frac{B}{NT_f}.
\end{equation}
The power for each frame $i$ is calculated as
\begin{equation}
p_i = \left[\frac{1}{\gamma_0}-\frac{1}{\gamma_i}\right]^+.
\label{eqn:pi}
\end{equation}
If $T_fp_i \le E_{Q,i}$, where $E_{Q,i}$ denotes the battery energy at
the beginning of frame $i$, the required power is supplied by the
battery only, and the energy from power grid is zero, i.e, $p_i =
p_{H,i}$. Otherwise, the amount of energy $T_fp_i - E_{Q,i}$ will be
taken from power grid.

\subsubsection{Adaptive water level}
For the finite time transmission, the constant water level algorithm
is apparently not optimal due to the different realization of
channel gains. We propose the adaptive water level algorithm to
improve the performance. The water level is updated for each frame, denoted by $\gamma_{0,i}$, which can be obtained by solving
\begin{equation}
\int_{\gamma_{0,i}}^{\infty}\frac{1}{2} \log_2 \left(
\frac{\gamma}{\gamma_{0,i}} \right) f(\gamma) d\gamma =
\frac{B_{\mathrm{R},i}}{(N-i+1)T_f},
\end{equation}
where $B_{\mathrm{R},i}$ is the remaining bits at the beginning of
frame $i$.

Due to the randomness of energy arrival, there may be energy battery
overflow for online solutions. Hence in addition, we propose an \emph{energy overflow
protection} algorithm. If the expected battery energy exceeds the
battery capacity, i.e, $E_{Q,i} + P_{H,\mathrm{ave}}T_f >
E_{\mathrm{max}}$, a minimum power
\begin{equation}
p_{\mathrm{min}} = \frac{E_{Q,i} - E_{\mathrm{max}}}{T_f} +
P_{H,\mathrm{ave}} \label{p:min}
\end{equation}
must be used. Then the power allocated in frame $i$ is determined as
\begin{equation}
\tilde{p}_i = \max\{p_i, p_{\mathrm{min}}\}, \label{eqn:poverflow}
\end{equation}
where $p_i$ is expressed as (\ref{eqn:pi}). In this way, a certain amount of power must be used even when the channel is in deep fading or the water level exceeds the optimal one. This algorithm well meets the battery
capacity constraint from the average point of view, and is expected to
improve the performance, although the battery overflow can not be
completely avoided.

Notice that all the packets need to be delivered in $N$ frames. Hence, in the last frame, all the remaining packets should be transmitted as long as the required transmit power is achievable. Specifically, the transmit power of the last is determined by
\begin{equation}
p_N = \min\left\{P_{T,\mathrm{max}}, \frac{2^{2B_{\mathrm{R},N}/T_f}-1}{\gamma_N}\right\},
\label{eqn:pN}
\end{equation}
where $P_{T,\mathrm{max}}$ is the maximum achievable transmit power. However for the online settings, the transmission completion can not be guaranteed due to the fading, the maximum transmit power limits, and so on. We assume that the packets remained in the transmitter by the end of the last frame are dropped, and define the \emph{packet dropping probability} as the ratio between the number of bits remained after the transmission and the total amount of bits arrived. Online algorithms should minimize the grid energy consumption as well as the packet dropping probability.

\subsection{Random packet arrival}

For the case that the bits arrive at the transmitter randomly, the online algorithms can be proposed as a direct extension of that for all packets ready before transmission case. Specifically, the constant water level algorithm is as follows. Assume the average packet arrival rate is $\bar{B}$ bits/frame, then the constant water level $1/\gamma_0$ can be obtained by
\begin{equation}
\int_{\gamma_0}^{\infty}\frac{1}{2} \log_2 \left(
\frac{\gamma}{\gamma_0} \right) f(\gamma) d\gamma =
\frac{\bar{B}}{T_f}.
\end{equation}
The difference with the previous one is that, as the available number of bits is finite, the power used in frame
$i$ should not exceed
\begin{equation}
p_{\mathrm{max}} = \min\left\{P_{T,\mathrm{max}}, \frac{2^{2B_{\mathrm{R},i}/T_f}-1}{\gamma_i}\right\}. \label{pi:ran}
\end{equation}
Hence the transmit power is expressed as $\tilde{p}_i = \min\left\{
p_i, p_{\mathrm{max}}\right\},$ where $p_i$ is expressed as (\ref{eqn:pi}).

The adaptive water level can be calculated as
\begin{equation}
\int_{\gamma_{0,i}}^{\infty}\frac{1}{2} \log_2 \left(
\frac{\gamma}{\gamma_{0,i}} \right) f(\gamma) d\gamma =
\frac{B_{\mathrm{R},i} + (N-i)\bar{B}}{(N-i+1)T_f},
\end{equation}
where the right side of the equation is the total number of remaining bits approximated as the sum
of current remaining bits $B_{\mathrm{R},i}$ and expected arrival in the upcoming
frames.

The battery overflow protection algorithm needs to be modified as
the number of available bits in each frame is random and finite.
In particular, the transmit power is determined by $p_i' = \min\left\{\max\{p_i, p_{\mathrm{min}}\}, p_{\mathrm{max}}\right\},$
where $p_{\mathrm{min}}$ is expressed as (\ref{p:min}).

\section{Numerical Results} \label{sec:simu}
In this section, we test the performance of the proposed algorithms.
We take the EARTH project simulation parameters for the
femto-cell scenario \cite{Earth:2010}. Specifically, we set the maximum transmit power as 33dBm for simulations. A total number of 1000 monte carlo simulation runs are performed for each parameter setup. For each
run, the channel reference SNRs and the energy arrivals are
randomly generated. The distribution of reference SNR follows the
Rayleigh model with probability density function
$p_{\gamma}(\gamma_i) =
\frac{1}{\bar{\gamma}}\exp\left(-\frac{\gamma_i}{\bar{\gamma}}\right)$,
where $\bar{\gamma}$ is the average reference SNR (set to be 0dB).
The harvested energy follows non-negative uniform distribution with
mean $P_{H,\mathrm{ave}} = 20$dBm.
The offline optimal solution, which is obtained by solving the
convex optimization problem assuming all the reference SNRs and the energy
arrivals are known before transmission, can be considered as the performance
upper bound.

\begin{figure}
\centering
\includegraphics[width=4.5in]{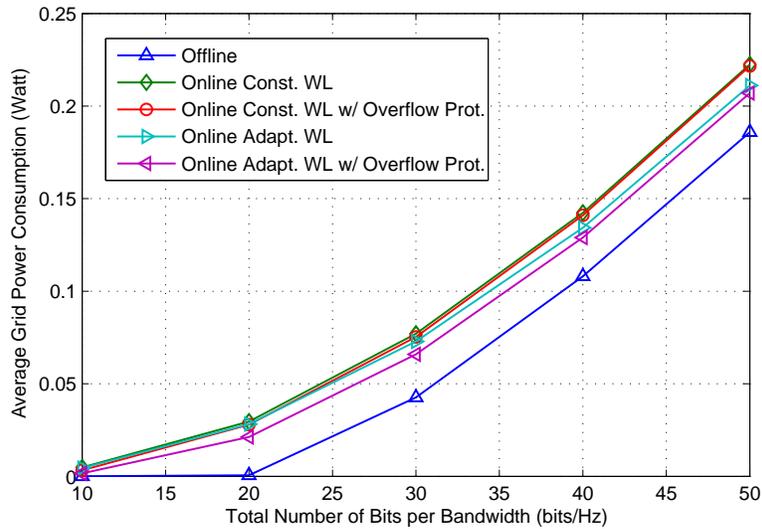}
\caption{Performance of average grid power consumption versus number of bits per bandwidth when packets are ready before transmission. $N = 100, P_{H,\mathrm{ave}} = 20$dBm, and the relative battery capacity $M=3$.}
\label{fig:VarBits}
\end{figure}

We first exam the performance of the algorithms for the all packets ready before transmission case. We
run $N=100$ frames in each simulation. The battery capacity is set $M$ times of the average arrived energy in one frame, i.e,
$E_{\mathrm{max}} = MP_{H,\mathrm{ave}}T_f$. The parameter $M$ is called \emph{relative battery capacity}. And we set $M=3$ in this setup. The average grid power consumption achieved versus the number of bits to be transmitted is shown in Fig.~\ref{fig:VarBits}. It shows that the adaptive water level algorithm outperforms the constant one. Because of the finite time length, the adaptive algorithm can better utilize the energy. Besides, the battery overflow protection algorithm improves the performance. It efficiently reduces the energy overflow at low grid power regime, but the gain shrinks as the grid power increases.

\begin{figure}
\centering
\includegraphics[width=4.5in]{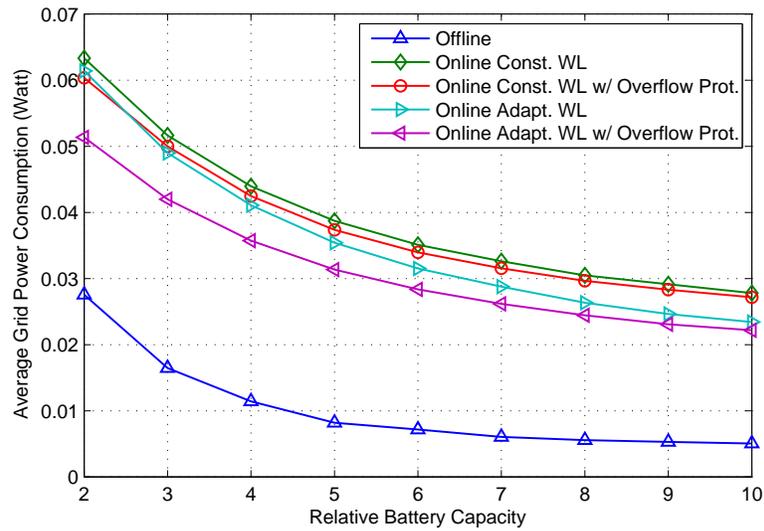}
\caption{Performance of average grid power consumption versus relative battery capacity when packets are ready before transmission. $N = 100, P_{H,\mathrm{ave}} = 20$dBm, and the number of bits per bandwidth is 25 bits/Hz.}
\label{fig:PminVarEmax}
\end{figure}

We also investigate the influence of renewable energy battery capacity on the
performance. Fig.~\ref{fig:PminVarEmax} shows the relation between the
battery capacity and the grid power consumption. Here, we set the number of bits per bandwidth as 25 bits/Hz. The performance
gap between the online algorithm and the offline algorithm becomes
small as the battery capacity increases. It shows that larger
battery capacity makes the power allocation more flexible. Also,
similar to the result from Fig. \ref{fig:VarBits}, the battery
overflow protection algorithm is efficient at low battery capacity
regime.

\begin{figure}
\centering
\includegraphics[width=4.5in]{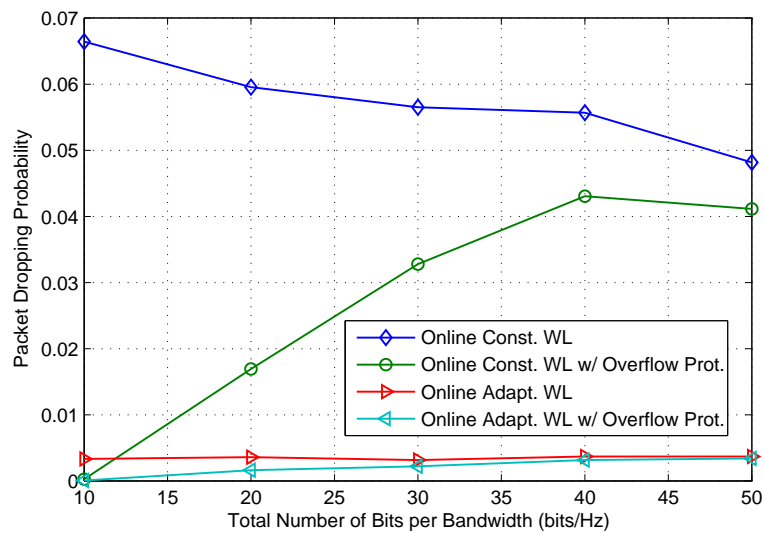}
\caption{Packet dropping probability versus number of bits per bandwidth when packets are ready before transmission. $N = 100, P_{H,\mathrm{ave}} = 20$dBm, and the relative battery capacity $M=3$.}
\label{fig:incomplete}
\end{figure}

Fig.~\ref{fig:incomplete} demonstrates the packet dropping probability for different packets transmission requirements. It can be found that the adaptive algorithm can effectively complete the transmission for almost all the conditions (the packet dropping probability is less than 0.4\%). On the contrary, the constant water level algorithm performs worse. Specifically, without overflow protection, the packet dropping probability keeps relatively high ($>4$\%), while with overflow protection, it is small for very few packets transmission requirement, and then increases and converges to that without overflow protection. Hence, the adaptive water level algorithm is an effective online algorithm.

\begin{figure}
\centering
\includegraphics[width=4.5in]{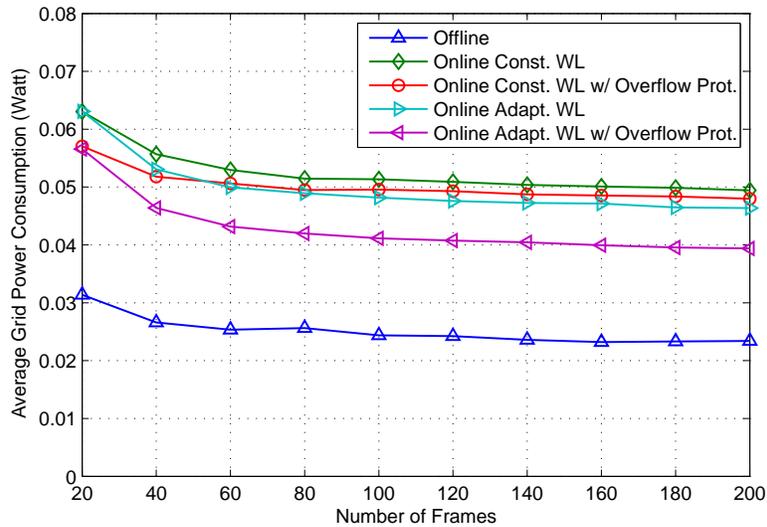}
\caption{Performance of average grid power consumption versus number of frames when packets are ready before transmission. $P_{H,\mathrm{ave}} = 20$dBm, $M = 3$, and the number of bits is 0.25 bits/Hz/frame.} \label{fig:VarFrame}
\end{figure}

In addition, the average grid power consumption versus the number of considered frames $N$ is depicted in Fig.~\ref{fig:VarFrame}. Generally speaking, the grid power consumption becomes stable when the number of frames gets large. It can be seen that $N\ge100$ is enough to get stable results, which validates our parameter setting.

\begin{figure}
\centering
\includegraphics[width=4.5in]{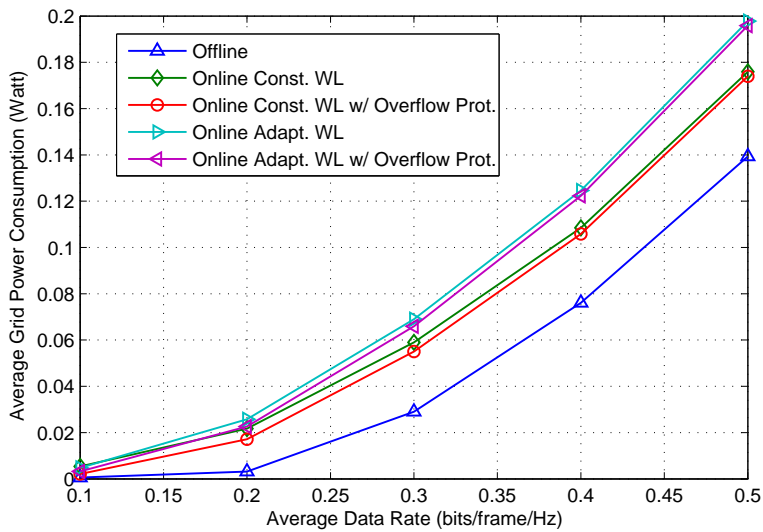}
\caption{Performance of average grid power consumption versus average data rate when packets randomly arrive during transmission. $N = 100, P_{H,\mathrm{ave}} = 20$dBm, and the relative battery capacity $M=3$.}
\label{fig:VarRanBit}
\end{figure}

\begin{figure}
\centering
\includegraphics[width=4.5in]{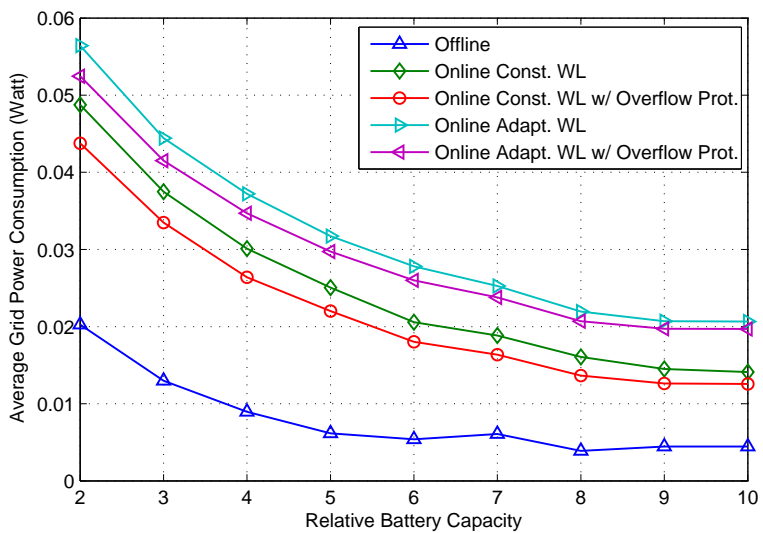}
\caption{Performance of average grid power consumption versus relative battery capacity in grid power minimization problem when packets randomly arrive during transmission. $N = 100, P_{H,\mathrm{ave}} = 20$dBm, and the average data rate is 0.25 bits/frame/Hz.}
\label{fig:VarRanEmax}
\end{figure}

\begin{figure}
\centering
\includegraphics[width=4.5in]{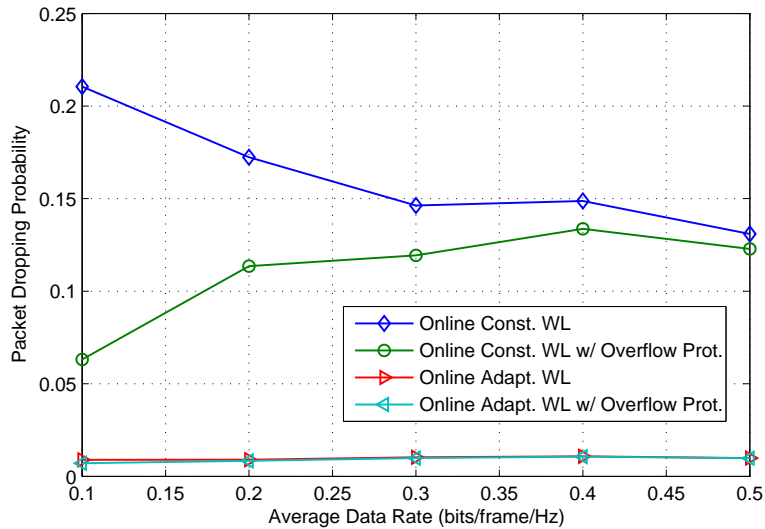}
\caption{Packet dropping probability versus average data rate when packets randomly arrive during transmission. $N = 100, P_{H,\mathrm{ave}} = 20$dBm, and the relative battery capacity $M=3$.}
\label{fig:incompleteRan}
\end{figure}

We then evaluate the performances of the algorithms for random packets arrival case. Figs.~\ref{fig:VarRanBit} and \ref{fig:VarRanEmax} show the grid power consumption versus average data arrival rate and relative battery capacity, respectively. In contrast to the previous results, we find that the constant water level algorithm with battery overflow protection achieves lower grid power consumption than the adaptive one. However, it is observed from Fig.~\ref{fig:incompleteRan} that the constant water level algorithm results in much higher packet dropping probability ($>12$\%) compared with the adaptive algorithm ($<1$\%). The constant water level algorithm reduces the grid power consumption by dropping a large amount of arrived bits, which can not well meet the data transmission requirement. Hence, in random packets arrival scenario, the adaptive water level algorithm requires more grid energy than the constant water level algorithm in order to keep a low packet dropping probability.

\section{Conclusion} \label{sec:conc}
We have analyzed the optimal power allocation structure of the energy harvesting and power grid coexisting systems. We consider the power grid energy minimization problem and solve it by convex optimization. The two-stage WF algorithm is proposed to solve the problem in the case that all packets are ready before transmission, and the multi-stage algorithm is proposed to solve it in the case of infinite battery capacity. In addition, online algorithms motivated by the offline solutions are also proposed and evaluated by numerical simulations. For all packets ready before transmission case, the adaptive water level algorithm outperforms the constant one for both grid power consumption and packet dropping probability. On the contrary, if the packets randomly arrive in each frame, constant water level algorithm achieves lower power consumption but causes much higher packet dropping probability.


\ifCLASSOPTIONcaptionsoff
  \newpage
\fi


\begin{thebibliography}{99}

\bibitem{Marsan:2009}
M.~A.~Marsan, L.~Chiaraviglio, D.~Ciullo, and M.~Meo, ``Optimal energy saving in cellular access networks," \emph{IEEE ICC, GreenComm.~Workshop}, Jun.~2009

\bibitem{Oh:2011}
E.~Oh, B.~Krishnamachari, X.~Liu, Z.~Niu, ``Toward dynamic energy-efficient operation of cellular network infrastructure," \emph{IEEE Commun.~Mag.}, Jun.~2011

\bibitem{David:2010}
D.~Valerdi, Q.~Zhu, K.~Exadaktylos, S.~Xia, M.~Arranz, R.~Liu, and D.~Xu, ``Intelligent energy managed service for green base stations," \emph{IEEE Globecom GreenComm.~Workshop}, Dec.~2010

\bibitem{Cui:2012}
Y.~Cui, K.~N.~Lau, and Y.~Wu, ``Delay-aware BS discontinuous
transmission control and user scheduling for energy harvesting
downlink coordinated MIMO systems," \emph{IEEE Trans. Signal Processing},
Vol.~60, No.~7, Jul.~2012

\bibitem{Fu:2003}
A.~C.~Fu, E.~Modiano, and J.~N.~Tsitsiklis, ``Optimal energy
allocation and admission control for communications satellites,"
\emph{IEEE/ACM Trans.~Networking}, Vol.~11, No.~3, Jun.~2003

\bibitem{Gatz:2010}
M.~Gatzianas, L.~Georgiadis, and L.~Tassiulas, ``Control of wireless
networks with rechargeable batteries," \emph{IEEE Trans.~Commun.}~Vol.~9,
No.~2, Feb.~2010

\bibitem{Shar:2010}
V.~Sharma, U.~Mukherji, V.~Joseph, and S.~Gupta, ``Optimal energy
management policies for energy harvesting sensor nodes," \emph{IEEE Trans.
Wireless Commun.}, Vol.~9, No.~4, Apr.~2010

\bibitem{Ho:2010}
C.~K.~Ho and R.~Zhang, ``Optimal energy allocation for wireless
communications powered by energy harvesters," \emph{Int.~Symposium
Inf.~Theory}, Austin, Texas, U.S.A., Jun. 13-18, 2010

\bibitem{Dimi:2005}
D.~P.~Bertsekas, \emph{Dynamic programming and optimal control}, 3rd
ed., Athena Scientific, Belmont, Massachusetts, 2005


\bibitem{Jing:2012}
J.~Yang and S.~Ulukus, ``Optimal packet scheduling in an energy
harvesting communication system,'' \emph{IEEE Trans.~Commun.}~Vol.~60,
No.~1, Jan.~2012


\bibitem{Tut:2012}
K.~Tutuncuoglu and A.~Yener, ``Optimum transmission policies for
battery limited energy harvesting nodes," \emph{IEEE
Trans.~Commun.}~Vol.~11, No.~3, Mar.~2012


\bibitem{Ozel:2011}
O.~Ozel, K.~Tutuncuoglu, J.~Yang, S.~Ulukus, and A.~Yener,
``Transmission with energy harvesting nodes in fading wireless
channels: optimal policies,'' \emph{IEEE J. Selected Area Commun.} Vol. 29,
No. 8, Sept. 2011


\bibitem{Biyik:2002}
E.~U.~Biyikoglu, B.~Prabhakar, and A.~El Gamal, ``Energy-efficient
transmission over a wireless link via lazy packet scheduling,"
\emph{IEEE/ACM Trans. Networking}, Vol.~10, No.~4, pp.~487-499, Aug.~2002

\bibitem{Chen:2007}
W.~Chen, M.~J.~Neely, and U.~Mitra, ``Energy efficient scheduling
with individual packet delay constraints: offline and online
results," \emph{Proc.~IEEE Infocom}, 2007

\bibitem{Zafer:2009}
M.~A.~Zafer and E.~Modiano, ``A calculus approach to
energy-efficient data transmission with quality-of-service
constraints," \emph{IEEE/ACM Trans. Networking}, Vol.~17, No.~3,
pp.~898-911, Jun.~2009

\bibitem{Boyd:2004}
S.~Boyd and L.~Vandenberghe, \emph{Convex optimization}, Cambridge
University, 2004

\bibitem{Gold:2005}
A.~Goldsmith, \emph{Wireless Communications}, Cambridge University,
2005

\bibitem{Earth:2010}
EARTH project deliverable, D2.3, ``Energy efficiency analysis of the
reference systems, areas of improvements and target breakdown." 2010


\end{thebibliography}
\end{document}